\documentclass[fleqn]{llncs}
\usepackage{latexsym}
\usepackage{amssymb}
\usepackage{stmaryrd}
\usepackage{graphicx}
\usepackage{color}
\usepackage{url}
\usepackage{phonetic}
\usepackage{amsmath}

\newcommand{\be}{\begin{enumerate}}
\newcommand{\ee}{\end{enumerate}}
\newcommand{\bi}{\begin{itemize}}
\newcommand{\ei}{\end{itemize}}
\newcommand{\bc}{\begin{center}}
\newcommand{\ec}{\end{center}}
\newcommand{\bsp}{\begin{sloppypar}}
\newcommand{\esp}{\end{sloppypar}}

\newtheorem{thm}{Theorem}[subsection]

\newtheorem{lem}[thm]{Lemma}
\newtheorem{prop}[thm]{Proposition}
\newtheorem{rem}[thm]{Remark}

\newcommand{\sglsp}{\ }
\newcommand{\dblsp}{\ \ }

\newcommand{\sB}{\mbox{$\cal B$}}
\newcommand{\sC}{\mbox{$\cal C$}}
\newcommand{\sD}{\mbox{$\cal D$}}
\newcommand{\sE}{\mbox{$\cal E$}}

\newcommand{\sM}{\mbox{$\cal M$}}

\newcommand{\sP}{\mbox{$\cal P$}}

\newcommand{\sT}{\mbox{$\cal T$}}

\newcommand{\sV}{\mbox{$\cal V$}}

\renewcommand{\phi}{\varphi}

\newcommand{\churchqe}{$\mbox{\sc ctt}_{\rm qe}$}
\newcommand{\churchuqe}{$\mbox{\sc ctt}_{\rm uqe}$}

\newcommand{\qzero}{${\cal Q}_0$}
\newcommand{\qzerou}{${\cal Q}^{\rm u}_{0}$}

\newcommand{\set}[1]{{\{ #1 \}}}
\newcommand{\sembrack}[1]{\llbracket#1\rrbracket}
\newcommand{\synbrack}[1]{\ulcorner#1\urcorner}

\newcommand{\mname}[1]{\mbox{\sf #1}}

\newcommand{\mdot}{\mathrel.}
\newcommand{\tarrow}{\rightarrow}
\newcommand{\LambdaApp}{\lambda\,}
\newcommand{\Neg}{\neg}

\renewcommand{\And}{\wedge}
\newcommand{\Implies}{\supset}
\newcommand{\Or}{\vee}
\newcommand{\Iff}{\equiv}

\newcommand{\ForallApp}{\forall\,}

\newcommand{\ForsomeApp}{\exists\,}

\newcommand{\IotaApp}{\mbox{\rm I}\,}
\newcommand{\IsDef}{\downarrow}
\newcommand{\IsUndef}{\uparrow}

\newcommand{\QuasiEqual}{\simeq}
\newcommand{\Undefined}{\bot}
\newcommand{\If}{\mname{if}}
\newcommand{\IsDefApp}{\!\IsDef}
\newcommand{\IsUndefApp}{\!\IsUndef}

\newcommand{\TRUE}{\mbox{{\sc t}}}
\newcommand{\FALSE}{\mbox{{\sc f}}}

\title{Theory Morphisms in Church's Type Theory with Quotation and
  Evaluation\thanks{Published in: H. Geuvers et al., eds,
    \emph{Intelligent Computer Mathematics (CICM 2017)}, \emph{Lecture
      Notes in Computer Science}, Vol.~10383, pp.~147--162, Springer,
    2017. The final publication is available at Springer via
    \texttt{http://dx.doi.org/10.1007/978-3-319-62075-6\_11}.  This
    research was supported by NSERC.}}

\author{William M. Farmer}

\institute{%
Computing and Software, McMaster University, Canada\\
\url{http://imps.mcmaster.ca/wmfarmer}\\[1.5ex]
23 July 2017
}

\begin{document}

\maketitle

\begin{abstract}
{\churchqe} is a version of Church's type theory with global quotation
and evaluation operators that is engineered to reason about the
interplay of syntax and semantics and to formalize syntax-based
mathematical algorithms.  {\churchuqe} is a variant of {\churchqe}
that admits undefined expressions, partial functions, and multiple
base types of individuals.  It is better suited than {\churchqe} as a
logic for building networks of theories connected by theory morphisms.
This paper presents the syntax and semantics of {\churchuqe}, defines
a notion of a theory morphism from one {\churchuqe} theory to another,
and gives two simple examples involving monoids that illustrate the
use of theory morphisms in {\churchuqe}.
\end{abstract}

\iffalse 


\fi

\section{Introduction}\label{sec:introduction}

A \emph{syntax-based mathematical algorithm (SBMA)}, such as a
symbolic differentiation algorithm, manipulates mathematical
expressions in a mathematically meaningful way.  Reasoning about SBMAs
requires reasoning about the relationship between how the expressions
are manipulated and what the manipulations mean mathematically.  We
argue in~\cite{Farmer13} that a logic with quotation and evaluation would
provide a global infrastructure for formalizing SBMAs and reasoning
about the interplay of syntax and semantics that is embodied in them.

\emph{Quotation} is a mechanism for referring to a syntactic value
(e.g., a syntax tree) that represents the syntactic structure of an
expression, while \emph{evaluation} is a mechanism for referring to
the value of the expression that a syntactic value represents.
Incorporating quotation and evaluation into a traditional logic like
first-order logic or simple type theory is tricky; there are several
challenging problems that the logic engineer must
overcome~\cite{Farmer13,FarmerArxiv16}.
{\churchqe}~\cite{FarmerArxiv16,Farmer16} is a version of Church's
type theory with global quotation and evaluation operators inspired by
the quote and eval operators in the Lisp programming language.  We
show in~\cite{FarmerArxiv16} that formula schemas and meaning formulas
for SBMAs can be expressed in {\churchqe} using quotation and
evaluation and that such schemas and meaning formulas can be
instantiated and proved within the proof system for {\churchqe}.

The \emph{little theories method}~\cite{FarmerEtAl92b} is an approach
for understanding and organizing mathematical knowledge as a
\emph{theory graph}~\cite{Kohlhase14} consisting of axiomatic
\emph{theories} as nodes and \emph{theory morphisms}\footnote{Theory
  morphisms are also known as \emph{immersions}, \emph{realizations},
  \emph{theory interpretations}, \emph{translations}, and
  \emph{views}.} as directed edges.  A theory consists of a
\emph{language} of expressions that denote mathematical values and a
set of \emph{axioms} that express in the language assumptions about
the values.  A theory morphism is a meaning-preserving mapping from
the formulas of one theory to the formulas of another theory.  Theory
morphisms serve as information conduits that enable definitions and
theorems to be passed from an abstract theory to many other more
concrete theories~\cite{BarwiseSeligman97}.

A \emph{biform theory}~\cite{CaretteFarmer08,Farmer07b} is a
combination of an axiomatic theory and an algorithmic theory (a
collection of algorithms that perform symbolic computations).  It
consists of a \emph{language} $L$ generated from a set of
\emph{symbols}, a set of \emph{transformers}, and a set of
\emph{axioms}.  The expressions of $L$ denote mathematical values that
include syntactic values representing the expressions of $L$.  The
transformers are SBMAs and other algorithms that implement functions
on the expressions of $L$ and are represented by symbols of $L$.  The
axioms are formulas of $L$ that express properties about the symbols
and transformers of the biform theory.  Unlike traditional logics,
{\churchqe} is well suited for formalizing biform theories.  Can the
little theories method be applied to biform theories formalized in
{\churchqe}?  This would require a definition of a theory morphism for
{\churchqe} theories.

Defining a notion of a theory morphism in a logic with quotation is
not as straightforward as in a logic without quotation due to the
following problem:

\bi

  \item[] \emph{Constant Interpretation Problem}.  Let $T_1$ and $T_2$
    be theories in a logic with a quotation operator
    $\synbrack{\cdot}$.  If a theory morphism $\Phi$ from $T_1$ to
    $T_2$ interprets two distinct constants $c$ and $c'$ in $T_1$ by a
    single constant $d$ in $T_2$, then $\Phi$ would map the true
    formula $\synbrack{c} \not= \synbrack{c'}$ of $T_1$ to the false
    formula $\synbrack{d} \not= \synbrack{d}$ of $T_2$, and hence
    $\Phi$ would not be meaning preserving.  Similarly, if $\Phi$
    interprets $c$ as an expression $e$ in $T_2$ that is not a
    constant, then $\Phi$ would map a true formula like
    $\mname{is-constant}(\synbrack{c})$ to the false formula
    $\mname{is-constant}(\synbrack{e})$.

\ei

\noindent
This paper defines a notion of a theory morphism that overcomes this
problem in {\churchuqe}, a variant of {\churchqe} that admits
undefined expressions, partial functions, and multiple base types of
individuals.  {\churchuqe} merges the machinery for quotation and
evaluation found in {\churchqe}~\cite{FarmerArxiv16} with the
machinery for undefinedness found in {\qzerou}~\cite{Farmer08a}.  Like
{\churchqe} and {\qzerou}, {\churchuqe} is based on
{\qzero}~\cite{Andrews02}, Peter Andrews' elegant version of Church's
type theory.  See~\cite{FarmerArxiv16} for references related to
{\churchuqe}.

{\churchuqe} is better suited than {\churchqe} as a logic for the
little theories method for two reasons.  First, it is often convenient
for a theory morphism from $T_1$ to $T_2$ to interpret different kinds
of values by values of different types.  Since {\churchqe} contains
only one base type of individuals, $\iota$, all individuals in a
theory $T_1$ must be interpreted by values of the same type in $T_2$.
Allowing multiple base types of individuals in {\churchuqe} eliminates
this restriction. Second, it is often useful to interpret a type
$\alpha$ in $T_1$ by a subset of the denotation of a type $\beta$ in
$T_2$.  As shown in~\cite{Farmer93}, this naturally leads to partial
functions on the type $\beta$.  {\churchuqe} has built-in support for
partial functions and undefinedness based on the traditional approach
to undefinedness~\cite{Farmer04}; {\churchqe} has no such built-in
support.\footnote{A logic without support for partial functions and
  undefinedness --- such as {\churchqe} or the logic of
  HOL~\cite{GordonMelham93} --- can interpret $\alpha$ by a type
  $\beta'$ that is isomorphic to a subset of $\beta$.  However, this
  approach is more complicated and farther from standard mathematics
  practice than interpreting $\alpha$ directly by a subset of
  $\beta$.}

The rest of the paper is organized as follows.  The syntax and
semantics of {\churchuqe} are presented in sections~\ref{sec:syntax}
and~\ref{sec:semantics}.  The notion of a theory morphism in
{\churchuqe} is defined in section~\ref{sec:morphisms}.
Section~\ref{sec:examples} contains two simple examples of theory
morphisms in {\churchuqe} involving monoids.  The paper concludes in
section~\ref{sec:conclusion} with a summary of the paper's results and
some brief remarks about constructing theory morphisms in an
implementation of {\churchuqe} and about future work.

The syntax and semantics of {\churchuqe} are presented as briefly as
possible.  The reader should consult~\cite{Farmer08a}
and~\cite{FarmerArxiv16} for a more in-depth discussion on the ideas
underlying the syntax and semantics in {\churchuqe}.  Due to limited
space, a proof system is not given in this paper for {\churchuqe}.  A
proof system for {\churchuqe} can be straightforwardly derived by
merging the proof systems for {\churchqe}~\cite{FarmerArxiv16} and
{\qzerou}~\cite{Farmer08a}.

\section{Syntax}\label{sec:syntax}

The syntax of {\churchuqe} is the same as the syntax of
{\churchqe}~\cite{FarmerArxiv16} except that (1)~the types include
denumerably many base types of individuals instead of just the single
$\iota$ type, (2)~the expressions include conditional expressions, and
(3)~the logical constants include constants for definite description
and exclude $\mname{is-expr}_{\epsilon \tarrow o}$ --- which we will
see is not needed since all constructions are ``proper'' in
{\churchuqe}.

\subsection{Types}

Let $\sB$ be a denumerable set of symbols that contains $\omicron$ and
$\epsilon$ .  A \emph{type} of {\churchuqe} is a string of symbols
defined inductively by the following formation rules:
\be

  \item \emph{Base type}: If $\alpha \in \sB$, then $\alpha$ is a
    type.

  \item \emph{Function type}: If $\alpha$ and $\beta$ are types, then
    $(\alpha \tarrow \beta)$ is a type.

\ee
Let $\sT$ denote the set of types of {\churchuqe}.  $\omicron$ and
$\epsilon$ are the \emph{logical base types} of {\churchuqe}.
$\alpha,\beta,\gamma, \ldots$ are syntactic variables ranging over
types.  When there is no loss of meaning, matching pairs of
parentheses in types may be omitted.  We assume that function type
formation associates to the right so that a type of the form $(\alpha
\tarrow (\beta \tarrow \gamma))$ may be written as $\alpha \tarrow
\beta \tarrow \gamma$.

\subsection{Expressions}\label{subsec:expressions}

A \emph{typed symbol} is a symbol with a subscript from $\sT$.  Let
$\sV$ be a set of typed symbols such that $\sV$ contains denumerably
many typed symbols with subscript~$\alpha$ for each $\alpha \in \sT$.
A \emph{variable of type $\alpha$} of {\churchuqe} is a member of
$\sV$ with subscript~$\alpha$.  $\textbf{f}_\alpha, \textbf{g}_\alpha,
\textbf{h}_\alpha, \textbf{u}_\alpha, \textbf{v}_\alpha,
\textbf{w}_\alpha,\textbf{x}_\alpha, \textbf{y}_\alpha,
\textbf{z}_\alpha,\ldots$ are syntactic variables ranging over
variables of type~$\alpha$.  We will assume that $f_\alpha, g_\alpha,
h_\alpha, u_\alpha, v_\alpha, w_\alpha, x_\alpha, y_\alpha,
z_\alpha,\ldots$ are actual variables of type~$\alpha$ of
{\churchuqe}.

Let $\sC$ be a set of typed symbols disjoint from $\sV$ that includes
the typed symbols in Table~\ref{tab:log-con}.  A \emph{constant of
  type~$\alpha$} of {\churchuqe} is a member of $\sC$ with
subscript~$\alpha$.  The typed symbols in Table~\ref{tab:log-con} are
the \emph{logical constants} of {\churchuqe}.  $\textbf{c}_\alpha,
\textbf{d}_\alpha, \ldots$ are syntactic variables ranging over
constants of type~$\alpha$.

\begin{table}[t]
\bc
\begin{tabular}{|ll|}
\hline
$\mname{=}_{\alpha \tarrow \alpha \tarrow o}$ 
& {\dblsp}for all $\alpha \in \sT$\\
$\iota_{(\alpha \tarrow o) \tarrow \alpha}$ 
& {\dblsp}for all $\alpha \in \sT$ with $\alpha \not= o$\\
$\mname{is-var}_{\epsilon \tarrow o}$
&\\
$\mname{is-var}_{\epsilon \tarrow o}^{\alpha}$
& {\dblsp}for all $\alpha \in \sT$\\
$\mname{is-con}_{\epsilon \tarrow o}$
&\\
$\mname{is-con}_{\epsilon \tarrow o}^{\alpha}$
& {\dblsp}for all $\alpha \in \sT$\\
$\mname{app}_{\epsilon \tarrow \epsilon \tarrow \epsilon}$
&\\
$\mname{abs}_{\epsilon \tarrow \epsilon \tarrow \epsilon}$
&\\
$\mname{cond}_{\epsilon \tarrow \epsilon \tarrow \epsilon \tarrow \epsilon}$
&\\
$\mname{quo}_{\epsilon \tarrow \epsilon}$
&\\

\iffalse
$\mname{is-expr}_{\epsilon \tarrow o}$
&\\
\fi

$\mname{is-expr}_{\epsilon \tarrow o}^{\alpha}$
& {\dblsp}for all $\alpha \in \sT$\\
$\sqsubset_{\epsilon \tarrow \epsilon \tarrow o}$
&\\
$\mname{is-free-in}_{\epsilon \tarrow \epsilon \tarrow o}$
&\\
\hline
\end{tabular}
\ec
\caption{Logical Constants}\label{tab:log-con}
\end{table}

An \emph{expression of type $\alpha$} of {\churchuqe} is a string of
symbols defined inductively by the formation rules below.
$\textbf{A}_\alpha, \textbf{B}_\alpha, \textbf{C}_\alpha, \ldots$ are
syntactic variables ranging over expressions of type $\alpha$.  An
expression is \emph{eval-free} if it is constructed using just the
first six formation rules.
\be

  \item \emph{Variable}: $\textbf{x}_\alpha$ is an expression of type
    $\alpha$.

  \item \emph{Constant}: $\textbf{c}_\alpha$ is an expression of type
    $\alpha$.

  \item \emph{Function application}: $(\textbf{F}_{\alpha \tarrow
    \beta} \, \textbf{A}_\alpha)$ is an expression of type $\beta$.

  \item \emph{Function abstraction}: $(\LambdaApp \textbf{x}_\alpha
    \mdot \textbf{B}_\beta)$ is an expression of type $\alpha \tarrow
    \beta$.

  \item \emph{Conditional}: $(\If \; \textbf{A}_o \; \textbf{B}_\alpha
    \; \textbf{C}_\alpha)$ is an expression of type $\alpha$.

  \item \emph{Quotation}: $\synbrack{\textbf{A}_\alpha}$ is an
    expression of type $\epsilon$ if $\textbf{A}_\alpha$ is eval-free.

  \item \emph{Evaluation}: $\sembrack{\textbf{A}_\epsilon}_{{\bf
      B}_\beta}$ is an expression of type $\beta$.

\ee 

\noindent
The purpose of the second argument $\textbf{B}_\beta$ in an evaluation
$\sembrack{\textbf{A}_\epsilon}_{{\bf B}_\beta}$ is to establish the
type of the evaluation.\footnote{It would be more natural for the
  second argument of an evaluation to be a type, but that would lead
  to an infinite family of evaluation operators, one for every type,
  since type variables are not available in {\churchuqe} (as well as
  in {\churchqe} and {\qzero}).}  A \emph{formula} is an expression
of type $o$.  A \emph{predicate} is an expression of a type of the
form $\alpha \tarrow o$.  When there is no loss of meaning, matching
pairs of parentheses in expressions may be omitted.  We assume that
function application formation associates to the left so that an
expression of the form $((\textbf{G}_{\alpha \tarrow \beta \tarrow
  \gamma} \, \textbf{A}_\alpha) \, \textbf{B}_\beta)$ may be written
as $\textbf{G}_{\alpha \tarrow \beta \tarrow \gamma} \,
\textbf{A}_\alpha \, \textbf{B}_\beta$.

\begin{rem}[Conditionals]\label{note:cond-wff}\em\bsp
We will see in the next section that $(\If \; \textbf{A}_o \;
\textbf{B}_\alpha \; \textbf{C}_\alpha)$ is a conditional expression
that is not strict with respect to undefinedness.  For instance, if
$\textbf{A}_o$ is true, then $(\If \; \textbf{A}_o \;
\textbf{B}_\alpha \;\textbf{C}_\alpha)$ denotes the value of
$\textbf{B}_\alpha$ even when $\textbf{C}_\alpha$ is undefined.  We
construct conditionals using an expression constructor instead of a
constant since constants always denote functions that are effectively
strict with respect to undefinedness.  We will use conditional
expressions to restrict the domain of a function. \esp
\end{rem}

An occurrence of a variable $\textbf{x}_\alpha$ in an eval-free
expression $\textbf{B}_\beta$ is \emph{bound} [\emph{free}] if (1) it
is not in a quotation and (2) it is [not] in a subexpression of
$\textbf{B}_\beta$ of the form $\LambdaApp \textbf{x}_\alpha \mdot
\textbf{C}_\gamma$.  An eval-free expression is \emph{closed} if no
free variables occur in it.

\iffalse

$\textbf{A}_\alpha$ is \emph{free for $\textbf{x}_\alpha$ in}
$\textbf{B}_\beta$ if no free occurrence of $\textbf{x}_\alpha$ in
$\textbf{B}_\beta$ is within a subexpression of $\textbf{B}_\beta$ of
the form $\LambdaApp \textbf{y}_\gamma \mdot \textbf{C}_\delta$ such
that $\textbf{y}_\gamma$ is free in $\textbf{A}_\alpha$.

\fi

\subsection{Constructions}

Let $\sE$ be the function mapping eval-free expressions to
expressions of type $\epsilon$ that is defined inductively as follows:

\be

  \item $\sE(\textbf{x}_\alpha) = \synbrack{\textbf{x}_\alpha}$.

  \item $\sE(\textbf{c}_\alpha) = \synbrack{\textbf{c}_\alpha}$.

  \item $\sE(\textbf{F}_{\alpha \tarrow \beta} \, \textbf{A}_\alpha) =
    \mname{app}_{\epsilon \tarrow \epsilon \tarrow \epsilon} \,
    \sE(\textbf{F}_{\alpha \tarrow \beta}) \, \sE(\textbf{A}_\alpha)$.

  \item $\sE(\LambdaApp \textbf{x}_\alpha \mdot \textbf{B}_\beta) =
    \mname{abs}_{\epsilon \tarrow \epsilon \tarrow \epsilon} \,
    \sE(\textbf{x}_\alpha) \, \sE(\textbf{B}_\beta)$.

  \item $\sE(\If \; \textbf{A}_o \; \textbf{B}_\alpha \;
    \textbf{C}_\alpha) = \mname{cond}_{\epsilon \tarrow \epsilon
      \tarrow \epsilon \tarrow \epsilon} \, \sE(\textbf{A}_o) \,
    \sE(\textbf{B}_\alpha) \, \sE(\textbf{C}_\alpha).$

  \item $\sE(\synbrack{\textbf{A}_\alpha}) = \mname{quo}_{\epsilon
    \tarrow \epsilon} \, \sE(\textbf{A}_\alpha)$.

\ee

\noindent
A \emph{construction} of {\churchuqe} is an expression in the range of
$\sE$.  $\sE$ is clearly injective.  When $\textbf{A}_\alpha$ is
eval-free, $\sE(\textbf{A}_\alpha)$ is a construction that represents
the syntactic structure of $\textbf{A}_\alpha$.  That is,
$\sE(\textbf{A}_\alpha)$ is a syntactic value that represents how
$\textbf{A}_\alpha$ is constructed as an expression.  In contrast to
{\churchqe}, the constructions of {\churchuqe} do not include
``improper constructions'' --- such as $\mname{app}_{\epsilon \tarrow
  \epsilon \tarrow \epsilon} \, \synbrack{\textbf{x}_\alpha} \,
\synbrack{\textbf{x}_\alpha}$ --- that do not represent the syntactic
structures of eval-free expressions.

The six kinds of eval-free expressions and the syntactic values that
represent their syntactic structures are given in
Table~\ref{tab:eval-free-exprs}.  

\begin{table}[p]
\bc
\begin{tabular}{|lll|}
\hline

\textbf{Kind}
& \textbf{Syntax}
& \textbf{Syntactic Value}\\

Variable \hspace*{15ex}
& $\textbf{x}_\alpha$  \hspace*{15ex}
& $\synbrack{\textbf{x}_\alpha}$\\

Constant
& $\textbf{c}_\alpha$
& $\synbrack{\textbf{c}_\alpha}$\\

Function application
& $\textbf{F}_{\alpha \tarrow \beta} \, \textbf{A}_\alpha$
& $\mname{app}_{\epsilon \tarrow \epsilon \tarrow \epsilon} \,
  \sE(\textbf{F}_{\alpha \tarrow \beta}) \, \sE(\textbf{A}_\alpha)$\\

Function abstraction
& $\LambdaApp \textbf{x}_\alpha \mdot \textbf{B}_\beta$
& $\mname{abs}_{\epsilon \tarrow \epsilon \tarrow \epsilon} \,
  \sE(\textbf{x}_\alpha) \, \sE(\textbf{B}_\beta)$\\

Conditional
& $(\If \; \textbf{A}_o \; \textbf{B}_\alpha \; \textbf{C}_\alpha)$
& $\mname{cond}_{\epsilon \tarrow \epsilon \tarrow \epsilon \tarrow \epsilon} 
  \, \sE(\textbf{A}_o) \, \sE(\textbf{B}_\alpha) \, \sE(\textbf{C}_\alpha).$\\

Quotation
& $\synbrack{\textbf{A}_\alpha}$
& $\mname{quo}_{\epsilon \tarrow \epsilon} \, \sE(\textbf{A}_\alpha)$\\

\hline
\end{tabular}
\ec
\caption{Six Kinds of Eval-Free Expressions}\label{tab:eval-free-exprs}
\end{table}

\subsection{Theories}

Let $\sB' \subseteq \sB$ and $\sC' \subseteq \sC$.  A type $\alpha$ of
{\churchuqe} is a \emph{$\sB'$-type} if each base type occurring in
$\alpha$ is a member of $\sB'$.  An expression $\textbf{A}_\alpha$ of
{\churchuqe} is a \emph{$(\sB',\sC')$-expression} if each base type
and constant occurring in $\textbf{A}_\alpha$ is a member of $\sB'$
and $\sC'$, respectively.  A \emph{language} of {\churchuqe} is the
set of all $(\sB',\sC')$-expressions for some $\sB' \subseteq \sB$ and
$\sC' \subseteq \sC$ such that $\sB'$ contains the logical base types
of {\churchuqe} (i.e., $o$ and $\iota$) and $\sC'$ contains the
logical constants of {\churchuqe}.  A \emph{theory} of {\churchuqe} is
a pair $T=(L,\Gamma)$ where $L$ is a language of {\churchuqe} and
$\Gamma$ is a set of formulas in $L$ (called the \emph{axioms} of
$T$).  $\textbf{A}_\alpha$ is an \emph{expression of a theory $T$} if
$\textbf{A}_\alpha \in L$.

\subsection{Definitions and Abbreviations} \label{subsec:definitions}

As in~\cite{FarmerArxiv16}, we introduce in Table~\ref{tab:defs}
several defined logical constants and abbreviations.
$({\textbf{A}_\alpha\IsDefApp})$ says that $\textbf{A}_\alpha$ is
defined, and similarly, $({\textbf{A}_\alpha\IsUndefApp})$ says that
$\textbf{A}_\alpha$ is undefined.  $\textbf{A}_\alpha \QuasiEqual
\textbf{B}_\alpha$ says that $\textbf{A}_\alpha$ and
$\textbf{B}_\alpha$ are \emph{quasi-equal}, i.e., that
$\textbf{A}_\alpha$ and $\textbf{B}_\alpha$ are either both defined
and equal or both undefined.  $\IotaApp \textbf{x}_\alpha \mdot
\textbf{A}_o$ is a \emph{definite description}.  It denotes the unique
$\textbf{x}_\alpha$ that satisfies $\textbf{A}_o$.  If there is no or
more than one such $\textbf{x}_\alpha$, it is undefined.  The defined
constant $\Undefined_\alpha$ is a canonical undefined expression of
type $\alpha$.

\iffalse 
The former includes constants for true and false, the propositional
connectives, and canonical undefined expressions.  The latter includes
notation for equality, the propositional connectives, universal and
existential quantification, $\sqsubset_{\epsilon \tarrow \epsilon
  \tarrow o}$ as an infix operator, a simplified notation for
evaluations, defined and undefined expressions, quasi-equality,
definite description, and conditional expressions.
\fi

\begin{table}[p]
\bc
\begin{tabular}{|lll|}
\hline

$(\textbf{A}_\alpha = \textbf{B}_\alpha)$
& {\dblsp}stands for{\dblsp}
& $=_{\alpha \tarrow \alpha \tarrow o} \, \textbf{A}_\alpha \, \textbf{B}_\alpha$.\\

$(\textbf{A}_o \Iff \textbf{B}_o)$ 
& {\dblsp}stands for{\dblsp} 
& $=_{o \tarrow o \tarrow o} \, \textbf{A}_o \, \textbf{B}_o$.\\

$T_o$ 
& {\dblsp}stands for{\dblsp}
& $=_{o \tarrow o \tarrow o} \; = \; =_{o \tarrow o \tarrow o}$.\\

$F_o$ 
& {\dblsp}stands for{\dblsp}
& $(\LambdaApp x_o \mdot T_o) = (\LambdaApp x_o \mdot x_o).$\\

$(\ForallApp \textbf{x}_\alpha \mdot \textbf{A}_o)$ 
& {\dblsp}stands for{\dblsp}
& $(\LambdaApp \textbf{x}_\alpha \mdot T_o) = (\LambdaApp \textbf{x}_\alpha \mdot \textbf{A}_o)$.\\

$\wedge_{o \tarrow o \tarrow o}$ 
& {\dblsp}stands for{\dblsp}
& $\LambdaApp x_o \mdot \LambdaApp y_o \mdot {}$\\
& 
& \hspace*{2ex}$((\LambdaApp g_{o \tarrow o \tarrow o} \mdot 
g_{o \tarrow o \tarrow o} \, T_o \, T_o) = {}$\\
&
& \hspace*{3ex}$(\LambdaApp g_{o \tarrow o \tarrow o} \mdot 
g_{o \tarrow o \tarrow o} \, x_o \, y_o)).$\\

$(\textbf{A}_o \And \textbf{B}_o)$ 
& {\dblsp}stands for{\dblsp}
& $\wedge_{o \tarrow o \tarrow o} \, \textbf{A}_o \, \textbf{B}_o$.\\

$\Implies_{o \tarrow o \tarrow o}$ 
& {\dblsp}stands for{\dblsp}
& $\LambdaApp x_o \mdot \LambdaApp y_o \mdot (x_o = (x_o \And y_o)).$\\ 

$(\textbf{A}_o \Implies \textbf{B}_o)$ 
& {\dblsp}stands for{\dblsp}
& ${\Implies_{o \tarrow o \tarrow o}} \, \textbf{A}_o \,\textbf{B}_o$.\\

$\Neg_{o \tarrow o}$ 
& {\dblsp}stands for{\dblsp}
& ${=_{o \tarrow o \tarrow o}} \, F_o$.\\

$(\Neg\textbf{A}_o)$ 
& {\dblsp}stands for{\dblsp}
& $\Neg_{o \tarrow o} \, \textbf{A}_o$.\\

$\vee_{o \tarrow o \tarrow o}$ 
& {\dblsp}stands for{\dblsp}
& $\LambdaApp x_o \mdot \LambdaApp y_o \mdot \Neg (\Neg x_o \And \Neg y_o).$\\

$(\textbf{A}_o \Or \textbf{B}_o)$ 
& {\dblsp}stands for{\dblsp}
& ${\vee_{o \tarrow o \tarrow o}} \, \textbf{A}_o \, \textbf{B}_o$.\\

$(\ForsomeApp \textbf{x}_\alpha \mdot \textbf{A}_o)$ 
& {\dblsp}stands for{\dblsp}
& $\Neg(\ForallApp \textbf{x}_\alpha \mdot \Neg\textbf{A}_o)$.\\

$(\textbf{A}_\alpha \not= \textbf{B}_\alpha)$ 
& {\dblsp}stands for{\dblsp} 
& $\Neg(\textbf{A}_\alpha = \textbf{B}_\alpha)$.\\

$\textbf{A}_\epsilon \sqsubset_{\epsilon \tarrow \epsilon \tarrow o} \textbf{B}_\epsilon$
& {\dblsp}stands for{\dblsp}
& $\sqsubset_{\epsilon \tarrow \epsilon \tarrow o} \, 
\textbf{A}_\epsilon \, \textbf{B}_\epsilon$.\\

$\sembrack{\textbf{A}_\epsilon}_\beta$ 
& {\dblsp}stands for{\dblsp}
& $\sembrack{\textbf{A}_\epsilon}_{{\bf B}_\beta}$.\\

$(\textbf{A}_\alpha\IsDefApp)$
& {\dblsp}stands for{\dblsp}
& $\textbf{A}_\alpha = \textbf{A}_\alpha$.\\

$(\textbf{A}_\alpha\IsUndefApp)$
& {\dblsp}stands for{\dblsp}
& $\Neg(\textbf{A}_\alpha\IsDefApp)$.\\

$(\textbf{A}_\alpha \QuasiEqual \textbf{B}_\alpha)$
& {\dblsp}stands for{\dblsp}
& $({\textbf{A}_\alpha\IsDefApp} \Or {\textbf{B}_\alpha\IsDefApp})
\Implies \textbf{A}_\alpha = \textbf{B}_\alpha$.\\

$(\IotaApp \textbf{x}_\alpha \mdot \textbf{A}_o)$
& {\dblsp}stands for{\dblsp}
& $\iota_{(\alpha \tarrow o) \tarrow \alpha} \,
  (\LambdaApp \textbf{x}_\alpha \mdot \textbf{A}_o)$
  {\sglsp} where $\alpha \not= o$.\\

$\Undefined_o$ 
& {\dblsp}stands for{\dblsp}
& $F_o$.\\

$\Undefined_\alpha$ 
& {\dblsp}stands for{\dblsp}
& $\IotaApp x_\alpha \mdot x_\alpha \not= x_\alpha$ {\sglsp} where $\alpha \not= o$.\\

\iffalse
$(\If \, \textbf{A}_o \, \textbf{B}_\alpha \, \textbf{C}_\alpha)$
& {\dblsp}stands for{\dblsp}
& $\IotaApp x_\alpha \mdot 
  ((\textbf{A}_o \Implies x_\alpha = \textbf{B}_\alpha) \And
  (\Neg\textbf{A}_o \Implies x_\alpha = \textbf{C}_\alpha))$\\
& & where $\textbf{A}_o$, $\textbf{B}_\alpha$, and $\textbf{C}_\alpha$ are eval-free and\\
& & $x_\alpha$ does
  not occur in $\textbf{A}_o$, $\textbf{B}_\alpha$, or $\textbf{C}_\alpha$.\\
\fi

\hline
\end{tabular}
\ec
\caption{Definitions and Abbreviations}\label{tab:defs}
\end{table}

\section{Semantics}\label{sec:semantics}

The semantics of {\churchuqe} is the same as the semantics of
{\churchqe} except that the former admits undefined expressions in
accordance with the traditional approach to
undefinedness~\cite{Farmer04}.  Two principal changes are made to the
{\churchqe} semantics: (1)~The notion of a general model is redefined
to include partial functions as well as total functions.  (2)~The
valuation function for expressions is made into a partial function
that assigns a value to an expression iff the expression is defined
according to the traditional approach.

\subsection{Frames}

A \emph{frame} of {\churchuqe} is a collection $\set{D_\alpha \;|\;
  \alpha \in \sT}$ of domains such that:

\be

  \item $D_o = \set{\TRUE,\FALSE}$, the set of standard \emph{truth
    values}.

  \item $D_\epsilon$ is the set of \emph{constructions} of
    {\churchuqe}.

  \item For $\alpha \in \sB$ with $\alpha \not\in \set{o,\epsilon}$,
    $D_\alpha$ is a nonempty set of values (called
    \emph{individuals}).

  \item For $\alpha, \beta \in \sT$, $D_{\alpha \tarrow \beta}$ is
    some set of \emph{total} functions from $D_\alpha$ to $D_\beta$ if
    $\beta = o$ and some set of \emph{partial and total} functions
    from $D_\alpha$ to $D_\beta$ if $\beta \not= o$.

\ee

\subsection{Interpretations}

An \emph{interpretation} of {\churchuqe} is a pair $(\set{D_\alpha
  \;|\; \alpha \in \sT},I)$ consisting of a frame and an
interpretation function $I$ that maps each constant in $\sC$ of type
$\alpha$ to an element of $D_\alpha$ such that:

\be

  \item For all $\alpha \in \sT$, $I(=_{\alpha \tarrow \alpha \tarrow
    o})$ is the total function $f \in D_{\alpha \tarrow \alpha \tarrow
    o}$ such that, for all $d_1,d_2 \in D_\alpha$, $f(d_1)(d_2) =
    \TRUE$ iff $d_1 = d_2$.

  \item For all $\alpha \in \sT$ with $\alpha \not= o$,
    $I(\iota_{(\alpha \tarrow o) \tarrow \alpha})$ is the partial
    function $f \in D_{(\alpha \tarrow o) \tarrow \alpha}$ such that,
    for all $d \in D_{\alpha \tarrow o}$, if the predicate $d$
    represents a singleton $\set{d'} \subseteq D_\alpha$, then $f(d) =
    d'$, and otherwise $f(d)$ is undefined.

  \item $I(\mname{is-var}_{\epsilon \tarrow o})$ the total function $f
    \in D_{\epsilon \tarrow o}$ such that, for all constructions
    $\textbf{A}_\epsilon \in D_\epsilon$, $f(\textbf{A}_\epsilon) =
    \TRUE$ iff $\textbf{A}_\epsilon = \synbrack{\textbf{x}_\alpha}$
    for some variable $\textbf{x}_\alpha \in \sV$ (where $\alpha$ can
    be any type).

  \item For all $\alpha \in \sT$, $I(\mname{is-var}_{\epsilon \tarrow
    o}^{\alpha})$ is the total function $f \in D_{\epsilon \tarrow o}$
    such that, for all constructions $\textbf{A}_\epsilon \in
    D_\epsilon$, $f(\textbf{A}_\epsilon) = \TRUE$ iff
    $\textbf{A}_\epsilon = \synbrack{\textbf{x}_\alpha}$ for some
    variable $\textbf{x}_\alpha \in \sV$.

  \item $I(\mname{is-con}_{\epsilon \tarrow o})$ is the total function
    $f \in D_{\epsilon \tarrow o}$ such that, for all constructions
    $\textbf{A}_\epsilon \in D_\epsilon$, $f(\textbf{A}_\epsilon) =
    \TRUE$ iff $\textbf{A}_\epsilon = \synbrack{\textbf{c}_\alpha}$
    for some constant $\textbf{c}_\alpha \in \sC$ (where $\alpha$ can
    be any type).

  \item For all $\alpha \in \sT$, $I(\mname{is-con}_{\epsilon \tarrow
    o}^{\alpha})$ is the total function $f \in D_{\epsilon \tarrow o}$
    such that, for all constructions $\textbf{A}_\epsilon \in
    D_\epsilon$, $f(\textbf{A}_\epsilon) = \TRUE$ iff
    $\textbf{A}_\epsilon = \synbrack{\textbf{c}_\alpha}$ for some
    constant $\textbf{c}_\alpha \in \sC$.

  \item $I(\mname{app}_{\epsilon \tarrow \epsilon \tarrow \epsilon})$
    is the partial function $f \in D_{\epsilon \tarrow \epsilon \tarrow
      \epsilon}$ such that, for all constructions
    $\textbf{A}_\epsilon, \textbf{B}_\epsilon \in D_\epsilon$, if
    $\mname{app}_{\epsilon \tarrow \epsilon \tarrow \epsilon} \,
    \textbf{A}_\epsilon \, \textbf{B}_\epsilon$ is a construction,
    then $f(\textbf{A}_\epsilon)(\textbf{B}_\epsilon) =
    \mname{app}_{\epsilon \tarrow \epsilon \tarrow \epsilon} \,
    \textbf{A}_\epsilon \, \textbf{B}_\epsilon$, and otherwise
    $f(\textbf{A}_\epsilon)(\textbf{B}_\epsilon)$ is undefined.

  \item $I(\mname{abs}_{\epsilon \tarrow \epsilon \tarrow \epsilon})$
    is the partial function $f \in D_{\epsilon \tarrow \epsilon \tarrow
      \epsilon}$ such that, for all constructions
    $\textbf{A}_\epsilon, \textbf{B}_\epsilon \in D_\epsilon$, if
    $\mname{abs}_{\epsilon \tarrow \epsilon \tarrow \epsilon} \,
    \textbf{A}_\epsilon \, \textbf{B}_\epsilon$ is a construction,
    then $f(\textbf{A}_\epsilon)(\textbf{B}_\epsilon)
    =\mname{abs}_{\epsilon \tarrow \epsilon \tarrow \epsilon} \,
    \textbf{A}_\epsilon \, \textbf{B}_\epsilon$, and otherwise
    $f(\textbf{A}_\epsilon)(\textbf{B}_\epsilon)$ is undefined.

  \item $I(\mname{cond}_{\epsilon \tarrow \epsilon \tarrow \epsilon
    \tarrow \epsilon})$ is the partial function $f \in D_{\epsilon
    \tarrow \epsilon \tarrow \epsilon \tarrow \epsilon}$ such that,
    for all constructions $\textbf{A}_\epsilon, \textbf{B}_\epsilon,
    \textbf{C}_\epsilon \in D_\epsilon$, if $\mname{cond}_{\epsilon
      \tarrow \epsilon \tarrow \epsilon \tarrow \epsilon} \,
    \textbf{A}_\epsilon \, \textbf{B}_\epsilon \, \textbf{C}_\epsilon$
    is a construction, then
    $f(\textbf{A}_\epsilon)(\textbf{B}_\epsilon)(\textbf{C}_\epsilon)
    =\mname{cond}_{\epsilon \tarrow \epsilon \tarrow \epsilon \tarrow
      \epsilon} \, \textbf{A}_\epsilon \, \textbf{B}_\epsilon \,
    \textbf{C}_\epsilon$, and otherwise
    $f(\textbf{A}_\epsilon)(\textbf{B}_\epsilon)(\textbf{C}_\epsilon)$
    is undefined.

  \item $I(\mname{quo}_{\epsilon \tarrow \epsilon})$ is the total
    function $f \in D_{\epsilon \tarrow \epsilon}$ such that, for all
    constructions $\textbf{A}_\epsilon \in D_\epsilon$,
    $f(\textbf{A}_\epsilon) = \mname{quo}_{\epsilon \tarrow \epsilon}
    \, \textbf{A}_\epsilon$.

\iffalse
  \item $I(\mname{is-expr}_{\epsilon \tarrow o})$ is the total
    function $f \in D_{\epsilon \tarrow o}$ such that, for all
    constructions $\textbf{A}_\epsilon \in D_\epsilon$,
    $f(\textbf{A}_\epsilon) = \TRUE$ iff $\textbf{A}_\epsilon =
    \sE(\textbf{B}_\alpha)$ for some (eval-free) expression
    $\textbf{B}_\alpha$ (where $\alpha$ can be any type).
\fi

  \item For all $\alpha \in \sT$, $I(\mname{is-expr}_{\epsilon \tarrow
    o}^{\alpha})$ is the total function $f \in D_{\epsilon \tarrow o}$
    such that, for all constructions $\textbf{A}_\epsilon \in
    D_\epsilon$, $f(\textbf{A}_\epsilon) = \TRUE$ iff
    $\textbf{A}_\epsilon = \sE(\textbf{B}_\alpha)$ for some
    (eval-free) expression $\textbf{B}_\alpha$.

  \item $I(\sqsubset_{\epsilon \tarrow \epsilon \tarrow o})$ is the
    total function $f \in D_{\epsilon \tarrow \epsilon \tarrow
      \epsilon}$ such that, for all constructions
    $\textbf{A}_\epsilon, \textbf{B}_\epsilon \in D_\epsilon$,
    $f(\textbf{A}_\epsilon)(\textbf{B}_\epsilon) = \TRUE$ iff
    $\textbf{A}_\epsilon$ is a proper subexpression of
    $\textbf{B}_\epsilon$.

  \item $I(\mname{is-free-in}_{\epsilon \tarrow \epsilon \tarrow o})$
    is the total function $f \in D_{\epsilon \tarrow \epsilon \tarrow
      \epsilon}$ such that, for all constructions
    $\textbf{A}_\epsilon, \textbf{B}_\epsilon \in D_\epsilon$,
    $f(\textbf{A}_\epsilon)(\textbf{B}_\epsilon) = \TRUE$ iff
    $\textbf{A}_\epsilon = \synbrack{\textbf{x}_\alpha}$ for some
    $\textbf{x}_\alpha \in \sV$ and $\textbf{x}_\alpha$ is free in the
    expression $\textbf{C}_\beta$ such that $\textbf{B}_\epsilon =
    \sE(\textbf{C}_\beta)$.

\ee

An \emph{assignment} into a frame $\set{D_\alpha \;|\; \alpha \in
  \sT}$ is a function $\phi$ whose domain is $\sV$ such that
$\phi(\textbf{x}_\alpha) \in D_\alpha$ for each $\textbf{x}_\alpha \in
\sV$.  Given an assignment $\phi$, $\textbf{x}_\alpha \in \sV$, and $d
\in D_\alpha$, let $\phi[\textbf{x}_\alpha \mapsto d]$ be the
assignment $\psi$ such that $\psi(\textbf{x}_\alpha) = d$ and
$\psi(\textbf{y}_\beta) = \phi(\textbf{y}_\beta)$ for all variables
$\textbf{y}_\beta$ distinct from $\textbf{x}_\alpha$.  For an
interpretation $\sM = (\set{D_\alpha \;|\; \alpha \in \sT}, I)$,
$\mname{assign}(\sM)$ is the set of assignments into the frame of
$\sM$.

\subsection{General Models} \label{subsec:gen-models}

An interpretation $\sM = (\set{D_\alpha \;|\; \alpha \in \sT), I}$ is
a \emph{general model} for {\churchuqe} if there is a partial binary
valuation function $V^{\cal M}$ such that, for all assignments $\phi
\in \mname{assign}(\sM)$ and expressions $\textbf{D}_\delta$, either
$V^{\cal M}_{\phi}(\textbf{D}_\delta) \in D_\delta$ or $V^{\cal
  M}_{\phi}(\textbf{D}_\delta)$ is undefined\footnote{We write
  $V^{\cal M}_{\phi}(\textbf{D}_\delta)$ instead of $V^{\cal
    M}(\phi,\textbf{D}_\delta)$.} and each of the following conditions
is satisfied:

\be

  \item Let $\textbf{D}_\delta \in \sV$.  Then $V^{\cal
    M}_{\phi}(\textbf{D}_\delta) = \phi(\textbf{D}_\delta)$.

  \item Let $\textbf{D}_\delta \in \sC$.  Then $V^{\cal
    M}_{\phi}(\textbf{D}_\delta) = I(\textbf{D}_\delta)$.

  \item Let $\textbf{D}_\delta$ be $\textbf{F}_{\alpha \tarrow \beta}
    \, \textbf{A}_\alpha$.  If $V^{\cal M}_{\phi}(\textbf{F}_{\alpha
      \tarrow \beta})$ is defined, $V^{\cal
      M}_{\phi}(\textbf{A}_\alpha)$ is defined, and the function
    $V^{\cal M}_{\phi}(\textbf{F}_{\alpha \tarrow \beta})$ is defined
    at the argument $V^{\cal M}_{\phi}(\textbf{A}_\alpha)$,
    then \[V^{\cal M}_{\phi}(\textbf{D}_\delta) = V^{\cal
      M}_{\phi}(\textbf{F}_{\alpha \tarrow \beta})(V^{\cal
      M}_{\phi}(\textbf{A}_\alpha)).\] Otherwise, $V^{\cal
      M}_{\phi}(\textbf{D}_\delta) = \FALSE$ if $\beta = o$ and
    $V^{\cal M}_{\phi}(\textbf{D}_\delta)$ is undefined if $\beta
    \not= o$.

  \item Let $\textbf{D}_\delta$ be $\LambdaApp \textbf{x}_\alpha \mdot
    \textbf{B}_\beta$.  Then $V^{\cal M}_{\phi}(\textbf{D}_\delta)$ is
    the (partial or total) function $f \in D_{\alpha \tarrow \beta}$
    such that, for each $d \in D_\alpha$, $f(d) = V^{\cal
      M}_{\phi[{\bf x}_\alpha \mapsto d]}(\textbf{B}_\beta)$ if
    $V^{\cal M}_{\phi[{\bf x}_\alpha \mapsto d]}(\textbf{B}_\beta)$ is
    defined and $f(d)$ is undefined if $V^{\cal M}_{\phi[{\bf
          x}_\alpha \mapsto d]}(\textbf{B}_\beta)$ is undefined.

  \item Let $\textbf{D}_\delta$ be $(\If \; \textbf{A}_o \;
    \textbf{B}_\alpha \; \textbf{C}_\alpha)$.  If $V^{\cal
      M}_{\phi}(\textbf{A}_o) = \TRUE$ and $V^{\cal
      M}_{\phi}(\textbf{B}_\alpha)$ is defined, then $V^{\cal
      M}_{\phi}(\textbf{D}_\delta) = V^{\cal
      M}_{\phi}(\textbf{B}_\alpha)$.  If $V^{\cal
      M}_{\phi}(\textbf{A}_o) = \FALSE$ and $V^{\cal
      M}_{\phi}(\textbf{C}_\alpha)$ is defined, then $V^{\cal
      M}_{\phi}(\textbf{D}_\delta) = V^{\cal
      M}_{\phi}(\textbf{C}_\alpha)$.  Otherwise, $V^{\cal
      M}_{\phi}(\textbf{D}_\delta)$ is undefined.

  \item Let $\textbf{D}_\delta$ be $\synbrack{\textbf{A}_\alpha}$.
    Then $V^{\cal M}_{\phi}(\textbf{D}_\delta) =
    \sE(\textbf{A}_\alpha)$.

  \item Let $\textbf{D}_\delta$ be
    $\sembrack{\textbf{A}_\epsilon}_\beta$.  If $V^{\cal
    M}_{\phi}(\mname{is-expr}_{\epsilon \tarrow o}^{\beta} \,
    \textbf{A}_\epsilon) = \TRUE$, then \[V^{\cal
      M}_{\phi}(\textbf{D}_\delta) = V^{\cal
      M}_{\phi}(\sE^{-1}(V^{\cal M}_{\phi}(\textbf{A}_\epsilon))).\]
    Otherwise, $V^{\cal
      M}_{\phi}(\textbf{D}_\delta) = \FALSE$ if $\beta = o$ and
    $V^{\cal M}_{\phi}(\textbf{D}_\delta)$ is undefined if $\beta
    \not= o$.

\ee 

\begin{prop} \label{prop:gen-models-exist}
General models for {\churchuqe} exist.
\end{prop}

\begin{proof}
The proof is similar to the proof of the analogous proposition
in~\cite{FarmerArxiv16}. \hfill $\Box$
\end{proof}

Other theorems about the semantics of {\churchuqe} are the same or
very similar to the theorems about the semantics of {\churchqe} given
in~\cite{FarmerArxiv16}.

Let $\sM$ be a general model for {\churchuqe}.  $\textbf{A}_o$ is
\emph{valid in $\sM$}, written $\sM \vDash \textbf{A}_o$, if $V^{\cal
  M}_{\phi}(\textbf{A}_o) = \TRUE$ for all $\phi \in
\mname{assign}(\sM)$.  $\textbf{A}_o$ is \emph{valid in {\churchuqe}},
written ${} \vDash \textbf{A}_o$, if $\textbf{A}_o$ is valid in every
general model for {\churchuqe}.  An expression $\textbf{B}_\beta$ is
\emph{semantically closed} if no variable ``is effective in'' it, i.e,
\[\vDash \ForallApp \textbf{y}_\alpha \mdot ((\LambdaApp
\textbf{x}_\alpha \mdot \textbf{B}_\beta) \, \textbf{y}_\alpha =
\textbf{B}_\beta)\] holds for all variables $\textbf{x}_\alpha$ (where
$\textbf{y}_\alpha$ is any variable of type $\alpha$ that differs from
$\textbf{x}_\alpha$).  It is easy to show that every closed eval-free
expression is semantically closed.  If $\textbf{B}_\beta$ is
semantically closed, then $V_{\phi}^{\cal M}(\textbf{B}_\beta)$ does
not depend on $\phi \in \mname{assign}(\sM)$.  The notion of
``$\textbf{x}_\alpha$ is effective in $\textbf{B}_\beta$'' is
discussed in detail in~\cite{FarmerArxiv16}.

Let $T=(L,\Gamma)$ be a theory of {\churchuqe} and $\textbf{A}_o$ be a
formula of $T$.  A \emph{general model for $T$} is a general model
$\sM$ for {\churchuqe} such that $\sM \vDash \textbf{A}_o$ for all
$\textbf{A}_o \in \Gamma$.  $\textbf{A}_o$ is \emph{valid in $T$},
written $T \vDash \textbf{A}_o$, if $\textbf{A}_o$ is valid in every
general model for $T$.  $T$ is \emph{normal} if each member of
$\Gamma$ is semantically closed.

\iffalse 

The proofs of the following two theorems are exactly the same as the
proofs of the corresponding theorems in~\cite{FarmerArxiv16}.

\begin{thm}[Law of Quotation] \label{thm:sem-quotation}
$\synbrack{\textbf{A}_\alpha} = \sE(\textbf{A}_\alpha)$ is valid in
  {\churchuqe}.
\end{thm}

\begin{thm}[Law of Disquotation] \label{thm:sem-disquotation}
$\sembrack{\synbrack{\textbf{A}_\alpha}}_\alpha = \textbf{A}_\alpha$
  is valid in {\churchuqe}.
\end{thm}

\fi % KEEP

\section{Theory Morphisms}\label{sec:morphisms}

In this section we define a ``semantic morphism'' of {\churchuqe} that
maps the valid semantically closed formulas of one normal theory to
the valid semantically closed formulas of another normal theory.
Theory morphisms usually map base types to types.  By exploiting the
support for partial functions in {\churchuqe}, we introduce a more
general notion of theory morphism that maps base types to semantically
closed predicates that represent sets of values of the same type.
This requires mapping expressions denoting functions on the base type
to expressions denoting functions with domains restricted to the
semantically closed predicate.

For $i =1,2$, let $T_i = (L_i,\Gamma_i)$ be a normal theory of
{\churchuqe} where, for some $\sB_i \subseteq \sB$ and $\sC_i
\subseteq \sC$, $L_i$ is the set of all $(\sB_i,\sC_i)$-expressions.
Also for $i =1,2$, let $\sT_i$ be the set of all $\sB_i$-types and
$\sV_i$ be the set of all variables in $L_i$.  Finally, let $\sP_2$ be
the set of all semantically closed predicates in $L_2$.

\subsection{Translations}

In this section, we will define a translation from $T_1$ to $T_2$ to
be a pair $(\mu,\nu)$ of functions where $\mu$ interprets the base
types of $T_1$ and $\nu$ interprets the variables and constants of
$T_1$.  $\overline{\mu}$ and $\overline{\nu}$ will be canonical
extensions of $\mu$ and $\nu$ to the types and expressions of $T_1$,
respectively.

Define $\tau$ to be the function that maps a predicate of type
$\alpha \tarrow o$ to the type~$\alpha$.  When $\textbf{p}_{\alpha
  \tarrow o}$ and $\textbf{q}_{\beta \tarrow o} $ are semantically
closed predicates, let
\[\textbf{p}_{\alpha \tarrow o} \rightharpoonup \textbf{q}_{\beta \tarrow o}\] 
be an abbreviation for the following semantically closed predicate of
type $(\alpha \tarrow \beta) \tarrow o$:
\[\LambdaApp f_{\alpha \tarrow \beta} \mdot 
\ForallApp x_\alpha \mdot (f_{\alpha \tarrow \beta} \, x_\alpha \not=
\Undefined_\beta \Implies (\textbf{p}_{\alpha \tarrow o} \, x_\alpha
\And \textbf{q}_{\beta \tarrow o} \, (f_{\alpha \tarrow \beta} \,
x_\alpha))).\] If $\beta = o$ [$\beta \not= o$], $\textbf{p}_{\alpha
  \tarrow o} \tarrow \textbf{q}_{\beta \tarrow o}$ represents the set
of total [partial and total] functions from the set of values
represented by $\textbf{p}_{\alpha \tarrow o}$ to the set of values
represented by $\textbf{q}_{\beta \tarrow o}$.  Notice
that \[\tau(\textbf{p}_{\alpha \tarrow o} \rightharpoonup
\textbf{q}_{\beta \tarrow o}) = \alpha \tarrow \beta =
\tau(\textbf{p}_{\alpha \tarrow o}) \tarrow \tau(\textbf{q}_{\beta
  \tarrow o}).\]

Given a total function $\mu : \sB_1 \tarrow \sP_2$, let
$\overline{\mu} : \sT_1 \tarrow \sP_2$ be the canonical extension of
$\mu$ that is defined inductively as follows:

\be

  \item If $\alpha \in \sB_1$, $\overline{\mu}(\alpha) = \mu(\alpha)$.

  \item If $\alpha \tarrow \beta \in \sT_1$, $\overline{\mu}(\alpha
    \tarrow \beta) = \overline{\mu}(\alpha) \rightharpoonup
    \overline{\mu}(\beta)$.

\ee

\noindent
It is easy to see that $\mu$ is well-defined and total.  

A \emph{translation from $T_1$ to $T_2$} is a pair $\Phi = (\mu,\nu)$,
where $\mu : \sB_1 \tarrow \sP_2$ is total and $\nu : \sV_1 \cup \sC_1
\tarrow \sV_2 \cup \sC_2$ is total and injective, such that:

\be

  \item $\mu(o) = \LambdaApp x_o \mdot T_o$.

  \item $\mu(\epsilon) = \LambdaApp x_\epsilon \mdot T_o$.

  \item For each $\textbf{x}_\alpha \in \sV_1$,
    $\nu(\textbf{x}_\alpha)$ is a variable in $\sV_2$ of type
    $\tau(\overline{\mu}(\alpha))$.

  \item For each $\textbf{c}_\alpha \in \sC_1$,
    $\nu(\textbf{c}_\alpha)$ is a constant in $\sC_2$ of type
    $\tau(\overline{\mu}(\alpha))$.

\ee

Throughout the rest of this section, let $\Phi = (\mu,\nu)$ be a
translation from $T_1$ to $T_2$. $\overline{\nu} : L_1 \tarrow L_2$ is
the canonical extension of $\nu$ defined inductively as follows:

\be

  \item If $\textbf{x}_\alpha \in \sV_1$,
    $\overline{\nu}(\textbf{x}_\alpha) = \nu(\textbf{x}_\alpha)$.

  \item If $\textbf{c}_\alpha \in \sC_1$,
    $\overline{\nu}(\textbf{c}_\alpha) = \nu(\textbf{c}_\alpha)$.

  \item If $\textbf{F}_{\alpha \tarrow \beta} \, \textbf{A}_\alpha \in
    L_1$, then $\overline{\nu}(\textbf{F}_{\alpha \tarrow \beta} \,
    \textbf{A}_\alpha) = \overline{\nu}(\textbf{F}_{\alpha \tarrow
      \beta}) \, \overline{\nu}(\textbf{A}_\alpha)$.

  \item If $\LambdaApp \textbf{x}_\alpha \mdot \textbf{B}_\beta \in
    L_1$, then $\overline{\nu}(\LambdaApp \textbf{x}_\alpha \mdot
    \textbf{B}_\beta) = {}$
    \[ \LambdaApp \overline{\nu}(\textbf{x}_\alpha)
    \mdot (\If \; (\overline{\mu}(\alpha) \,
    \overline{\nu}(\textbf{x}_\alpha)) \;
    \overline{\nu}(\textbf{B}_\beta) \;
    \Undefined_{\tau(\overline{\mu}(\beta))}).\]

  \item If $(\If \; \textbf{A}_o \; \textbf{B}_\alpha \;
    \textbf{C}_\alpha) \in L_1$, $\overline{\nu}(\If \; \textbf{A}_o
    \; \textbf{B}_\alpha \; \textbf{C}_\alpha) = (\If \;
    \overline{\nu}(\textbf{A}_o) \; \overline{\nu}(\textbf{B}_\alpha)
    \; \overline{\nu}(\textbf{C}_\alpha))$.

  \item If $\synbrack{\textbf{A}_\alpha} \in L_1$, then
    $\overline{\nu}(\synbrack{\textbf{A}_\alpha}) =
    \synbrack{\overline{\nu}(\textbf{A}_\alpha)}$.

  \item If $\sembrack{\textbf{A}_\epsilon}_{{\bf B}_\beta} \in L_1$,
    then $\overline{\nu}(\sembrack{\textbf{A}_\epsilon}_{{\bf
        B}_\beta}) =
    \sembrack{\overline{\nu}(\textbf{A}_\epsilon)}_{\overline{\nu}({\bf
        B}_\beta)}$.

\ee

\begin{lem}\label{lem:translation}

\be

  \item[]

  \item $\overline{\nu}$ is well-defined, total, and injective.

  \item If $\textbf{A}_\alpha \in L_1$, then
    $\overline{\nu}(\textbf{A}_\alpha)$ is an expression of type
    $\tau(\overline{\mu}(\alpha))$.

\ee
\end{lem}

\begin{proof} 
The two parts of the proposition are easily proved simultaneously by
induction on the structure of expressions. \hfill $\Box$
\end{proof}

\begin{rem}\em
We overcome the Constant Interpretation Problem mentioned in
section~\ref{sec:introduction} by requiring $\nu$ to injectively map
constants to constants which, by Lemma~\ref{lem:translation}, implies
that $\overline{\nu}$ injectively maps expressions to expressions.  We
will see in the next section that this requirement comes with a cost.
\end{rem}

A formula in $L_2$ is an \emph{obligation} of $\Phi$ if
it is one of the following formulas:

\be

  \item $\ForsomeApp x_{\tau(\mu(\alpha))} \mdot \mu(\alpha) \,
    x_{\tau(\mu(\alpha))}$ {\sglsp} where $\alpha \in \sB_1$.

  \item $\overline{\mu}(\alpha) \,
    \nu(\textbf{c}_\alpha)$ {\sglsp} where $\textbf{c}_\alpha \in
    \sC_1$.

  \item $\nu(=_{\alpha \tarrow \alpha \tarrow o}) = 
    \LambdaApp x_{\alpha'} \mdot \LambdaApp y_{\alpha'} \mdot
    (\If \; (\overline{\mu}(\alpha) \, x_{\alpha'} \And 
    \overline{\mu}(\alpha) \, y_{\alpha'}) \; 
    (x_{\alpha'} =_{\alpha' \tarrow \alpha' \tarrow o} y_{\alpha'}) \; 
    \Undefined_o)$\\
    where $\alpha \in \sT_1$ and $\alpha' = \tau(\overline{\mu}(\alpha))$.

  \item $\nu(\iota_{(\alpha \tarrow o) \tarrow \alpha}) =
    \LambdaApp x_{\alpha' \tarrow o} \mdot 
    (\If \; (\overline{\mu}(\alpha \tarrow o) \, x_{\alpha' \tarrow o}) \;
    (\iota_{(\alpha' \tarrow o) \tarrow \alpha'} \, x_{\alpha' \tarrow o}) \;
    \Undefined_{\alpha'})$\\
    where $\alpha \in \sT_1$ with $\alpha \not = o$ and 
    $\alpha' = \tau(\overline{\mu}(\alpha))$.

  \item \bsp $\nu(\textbf{c}_\alpha) = \textbf{c}_\alpha$ {\sglsp}
    where $\textbf{c}_\alpha$ is $\mname{is-var}_{\epsilon \tarrow
      o}$, $\mname{is-con}_{\epsilon \tarrow o}$,
    $\mname{app}_{\epsilon \tarrow \epsilon \tarrow \epsilon}$,
    $\mname{abs}_{\epsilon \tarrow \epsilon \tarrow \epsilon}$,
    $\mname{cond}_{\epsilon \tarrow \epsilon \tarrow \epsilon \tarrow \epsilon}$
    $\mname{quo}_{\epsilon \tarrow \epsilon}$,
    $\sqsubset_{\epsilon
      \tarrow \epsilon \tarrow o}$, or $\mname{is-free-in}_{\epsilon
      \tarrow \epsilon \tarrow o}$.\esp

\iffalse
$\mname{is-expr}_{\epsilon \tarrow o}$, 
\fi

  \item $\nu(\textbf{c}_{\alpha}^{\beta}) =
    \textbf{c}_{\alpha}^{\tau(\overline{\mu}(\beta))}$ {\sglsp} where
    $\textbf{c}_\alpha$ is $\mname{is-var}_{\epsilon \tarrow
      o}^{\beta}$, $\mname{is-con}_{\epsilon \tarrow o}^{\beta}$, or
    $\mname{is-expr}_{\epsilon \tarrow o}^{\beta}$ and $\beta \in
    \sT_1$.

  \item $\overline{\nu}(\textbf{A}_o)$ {\sglsp} where
    $\textbf{A}_o \in \Gamma_1$.

\ee

\noindent
Notice that each obligation of $\Phi$ is semantically closed.

\subsection{Semantic Morphisms}

A \emph{semantic morphism} from $T_1$ to $T_2$ is a translation
$(\mu,\nu)$ from $T_1$ to $T_2$ such that $T_1 \vDash \textbf{A}_o$
implies $T_2 \vDash \overline{\nu}(\textbf{A}_o)$ for all semantically
closed formulas $\textbf{A}_o$ of $T_1$.  (A \emph{syntactic morphism}
from $T_1$ to $T_2$ would be a translation $(\mu,\nu)$ from $T_1$ to
$T_2$ such that $T_1 \vdash_P \textbf{A}_o$ implies $T_2 \vdash_P
\overline{\nu}(\textbf{A}_o)$ for all semantically closed formulas
$\textbf{A}_o$ of $T_1$ where $P$ is some proof system for
         {\churchuqe}.)  We will prove a theorem (called the Semantic
         Morphism Theorem) that gives a sufficient condition for a
         translation to be a semantic morphism.

Assume $\sM_2 = (\set{D_{\alpha}^{2} \;|\; \alpha \in \sT},I_2)$ is a
general model for $T_2$.  Under the assumption that the obligations of
$\Phi$ are valid in $T_2$, we will extract a general model for $T_1$
from $\sM_2$.

For each $\alpha \in \sT_1$, define
$\underline{D}_{\tau(\overline{\mu}(\alpha))}^{2} \subseteq
D_{\tau(\overline{\mu}(\alpha))}^{2}$ as follows:

\be

  \item $\underline{D}_{\tau(\overline{\mu}(o))}^{2} =
    \underline{D}_{o}^{2} = D_{o}^{2} =
    \set{\TRUE,\FALSE}$.

  \item $\underline{D}_{\tau(\overline{\mu}(\epsilon))}^{2} =
    \underline{D}_{\epsilon}^{2} = {}$ \[\set{d \in
      D_{\epsilon}^{2}\;|\; d = V_{\phi}^{{\cal
          M}_2}(\overline{\nu}(\textbf{A}_\epsilon)) \mbox{ for some
        construction } \textbf{A}_\epsilon \in L_1}\] where $\phi$ is
    any member of $\mname{assign}(\sM_2)$.

  \item If $\alpha \in \sT_1 \setminus \set{o,\epsilon}$,
    $\underline{D}_{\tau(\overline{\mu}(\alpha))}^{2} = {}$\\ \[\set{d
    \in D_{\tau(\overline{\mu}(\alpha))}^{2} \;|\; V_{\phi}^{{\cal
        M}_2}(\overline{\mu}(\alpha))(d) = \TRUE}\] where $\phi$ is
    any member of $\mname{assign}(\sM_2)$.

\ee

For each $\alpha \in \sT_1$, define $\overline{D}_{\alpha}^{1}$
inductively as follows:

\be

  \item $\overline{D}_{o}^{1} = \set{\TRUE,\FALSE}$.

  \item $\overline{D}_{\epsilon}^{1}$ is the set of constructions of
    {\churchuqe}.

  \item If $\alpha \in \sB_1 \setminus \set{o,\epsilon}$,
    $\overline{D}_{\alpha}^{1} =
    \underline{D}_{\tau(\overline{\mu}(\alpha))}^{2}$.

  \item If $\alpha \tarrow \beta \in \sT_1$, then $\overline{D}_{\alpha
    \tarrow \beta}^{1}$ is the set of all \emph{total} functions from
    $\overline{D}_{\alpha}^{1}$ to $\overline{D}_{\beta}^{1}$ if
    $\beta = o$ and the set of all \emph{partial and total} functions
    from $\overline{D}_{\alpha}^{1}$ to $\overline{D}_{\beta}^{1}$ if
    $\beta \not= o$.

\ee

For each $\alpha \in \sT_1$, define $\rho_\alpha :
\underline{D}_{\tau(\overline{\mu}(\alpha))}^{2} \tarrow \overline{D}_{\alpha}^{1}$
inductively as follows:

\be

  \item If $d \in \underline{D}_{\epsilon}^{2}$, $\rho_\epsilon(d)$ is
    the unique construction $\textbf{A}_\epsilon$ such that
    $\overline{\nu}(\textbf{A}_\epsilon) = d$.

  \item If $\alpha \in \sB_1 \setminus \set{\epsilon}$ and $d \in
    \underline{D}_{\tau(\overline{\mu}(\alpha))}^{2}$, $\rho_\alpha(d)
    = d$.

  \item If $\alpha \tarrow \beta \in \sT_1$ and $f \in
    \underline{D}_{\tau(\overline{\mu}(\alpha \tarrow \beta))}^{2}$,
    $\rho_{\alpha \tarrow \beta}(f)$ is the unique function $g \in
    \overline{D}_{\alpha \tarrow \beta}^{1}$ such that, for all $d \in
    \underline{D}_{\tau(\overline{\mu}(\alpha))}^{2}$, either $f(d)$
    and $g(\rho_{\alpha}(d))$ are both defined and $\rho_{\beta}(f(d))
    = g(\rho_{\alpha}(d))$ or they are both undefined.

\ee

\begin{lem}
If $\alpha \in \sT_1$, $\rho_\alpha :
\underline{D}_{\tau(\overline{\mu}(\alpha))}^{2} \tarrow
\overline{D}_{\alpha}^{1}$ is well defined, total, and injective.
\end{lem}

\begin{proof}
This lemma is proved by induction on $\alpha \in \sT_1$.
$\rho_\epsilon$ is well defined since $V_{\phi}^{{\cal M}_2}$ is
identity function on constructions and $\overline{\nu}$ is injective
by Lemma~\ref{lem:translation}. \hfill $\Box$
\end{proof}

For each $\alpha \in \sT_1$, define $D_{\alpha}^{1} \subseteq
\overline{D}_{\alpha}^{1}$ as follows:

\be

  \item If $\alpha \in \sB_1$, $D_{\alpha}^{1} =
    \overline{D}_{\alpha}^{1}$.

  \item If $\alpha \tarrow \beta \in \sT_1$, $D_{\alpha \tarrow
    \beta}^{1}$ is the range of $\rho_{\alpha \tarrow \beta}$.

  \item If $\alpha \in \sB \setminus \sB_1$, $D_{\alpha}^{1}$ is any
    nonempty set.

  \item If $\alpha \tarrow \beta \in \sT \setminus \sT_1$,
    $D_{\alpha}^{1}$ is the set of all \emph{total} functions from
    $D_{\alpha}^{1}$ to $D_{\beta}^{1}$ if $\beta = o$ and the set of
    all \emph{partial and total} functions from $D_{\alpha}^{1}$ to
    $D_{\beta}^{1}$ if $\beta \not= o$.

\ee

\noindent
For $\textbf{c}_\alpha \in \sC_1$, define $I_1(\textbf{c}_\alpha) =
\rho_\alpha(V_{\phi}^{{\cal M}_2}(\overline{\nu}(\textbf{c}_\alpha)))$
where $\phi$ is any member of $\mname{assign}(\sM_2)$.  Finally,
define $\sM_1 = (\set{\sD_{\alpha}^{1} \;|\; \alpha \in \sT},I_1)$.

\begin{lem}\label{lem:is-model}
\bsp Suppose each obligation of $\Phi$ is valid in $\sM_2$.  Then
$\sM_1$ is a general model for $T_2$.  \esp
\end{lem}

\begin{proof}
By the first group of obligations of $\Phi$, $\sD_{\alpha}^{1}$ is
nonempty for all $\alpha \in \sB_1$, and so $\set{\sD_{\alpha}^{1}
  \;|\; \alpha \in \sT}$ is a frame of {\churchuqe}.  By the second to
sixth groups of obligations of $\Phi$, $\sM_1$ is an interpretation of
{\churchuqe}.  For all $\textbf{A}_\alpha \in L_1$ and $\phi \in
\mname{assign}(\sM_1)$, define $V_{\phi}^{{\cal
    M}_1}(\textbf{A}_\alpha)$ as follows:

\be

  \item[] $(\star)$ $V_{\phi}^{{\cal M}_1}(\textbf{A}_\alpha) =
    \rho_\alpha(V_{\overline{\nu}(\phi)}^{{\cal
        M}_2}(\overline{\nu}(\textbf{A}_\alpha)))$ if
    $V_{\overline{\nu}(\phi)}^{{\cal
        M}_2}(\overline{\nu}(\textbf{A}_\alpha))$ is defined and
    $V_{\phi}^{{\cal M}_1}(\textbf{A}_\alpha)$ is undefined otherwise,

\ee 
where $\overline{\nu}(\phi)$ is any $\psi \in \mname{assign}(\sM_2)$
such that, for all $\textbf{x}_\beta \in \sV_1$,
$\rho_\beta(\psi(\overline{\nu}(\textbf{x}_\beta))) =
\phi(\textbf{x}_\beta)$.  This definition of $V_{\phi}^{{\cal M}_1}$
can be easily extended to a valuation function on all expressions that
can be shown, by induction on the structure of expressions, to satisfy
the seven clauses of the definition of a general model.  Therefore,
$\sM_1$ is a general model for {\churchuqe}.  Then $(\star)$ implies
\[(\star\star) {\sglsp} \sM_1 \vDash \textbf{A}_o {\sglsp}
\mbox{iff} {\sglsp} \sM_2 \vDash \overline{\nu}(\textbf{A}_o)\] for
all semantically closed formulas $\textbf{A}_o \in L_1$.  By the
seventh group of obligations of $\Phi$, $\sM_2 \vDash
\overline{\nu}(\textbf{A}_o)$ for all $\textbf{A}_o \in \Gamma_1$, and
thus $\sM_1$ is a general model for $T_1$ by~$(\star\star)$. 

\hfill $\Box$
\end{proof}

\begin{thm}[Semantic Morphism Theorem]\label{thm:sem-morph}
Let $T_1$ and $T_2$ be normal theories and $\Phi$ be a translation
from $T_1$ to $T_2$.  Suppose each obligation of $\Phi$ is valid in
$T_2$.  Then $\Phi$ is a semantic morphism from $T_1$ to $T_2$.
\end{thm}

\begin{proof}
Let $\Phi= (\mu,\nu)$ be a translation from $T_1$ to $T_2$ and suppose
each obligation of $\Phi$ is valid in $T_2$.  Let $\textbf{A}_o \in
L_1$ be semantically closed and valid in $T_1$.  We must show that
$\overline{\nu}(\textbf{A}_o)$ is valid in every general model for
$T_2$.  Let $\sM_2$ be a general model for $T_1$.  (We are done if
there are no general models for $T_2$.)  Let $\sM_1$ be extracted from
$\sM_2$ as above. Obviously, each obligation of $\Phi$ is valid in
$\sM_2$, and so $\sM_1$ is a general model for $T_1$ by
Lemma~\ref{lem:is-model}.  Therefore, $\sM_1 \vDash \textbf{A}_o$ ,
and so $\sM_2 \vDash \overline{\nu}(\textbf{A}_o)$ by $(\star\star)$ in
the proof of Lemma~\ref{lem:is-model}. \hfill $\Box$
\end{proof}

\begin{thm}[Relative Satisfiability]
Let $T_1$ and $T_2$ be normal theories and suppose $\Phi$ is a
semantic morphism of from $T_1$ to $T_2$.  Then there is a general
model for $T_1$ if there is a general model for $T_2$.
\end{thm}

\begin{proof}
Let $\Phi= (\mu,\nu)$ be a semantic morphism from $T_1$ to $T_2$,
$\sM_2$ be a general model for $T_1$, and $\sM_1$ be extracted from
$\sM_2$ as above.  Since $\Phi$ is a semantic morphism, each of its
obligations is valid in $T_2$.  Hence, $\sM_1$ is a general model for
$T_1$ by Lemma~\ref{lem:is-model}. \hfill $\Box$
\end{proof}

\section{Examples}\label{sec:examples}

We will illustrate the theory morphism machinery of {\churchuqe} with
two simple examples involving monoids, the first in which two concepts
are interpreted as the same concept and second in which a type is
interpreted as a subset of its denotation.  Let $\sC_{\rm log}
\subseteq \sC$ be the set of logical constants of {\churchuqe}.

\subsection{Example 1: Monoid with Left and Right Identity Elements}

Define $M = (L_M,\Gamma_M)$ to be the usual theory of an abstract
monoid where:

\be

  \item $\sB_M = \set{o,\epsilon,\iota}$.

  \item $\sC_M = \sC_{\rm log} \cup \set{e_\iota, *_{\iota \tarrow
      \iota \tarrow \iota}}$. {\sglsp} ($*_{\iota \tarrow \iota \tarrow
    \iota}$ is written as an infix operator.)

  \item $L_M$ is the set of $(\sB_M,\sC_M)$ expressions.

  \item $\sV_M$ is the set of variables in $L_M$.

  \item $\Gamma_M$ contains the following axioms:

  \be

    \item $\ForallApp x_\iota \mdot \ForallApp y_\iota \mdot
      \ForallApp z_\iota \mdot x_\iota *_{\iota \tarrow \iota \tarrow
        \iota} (y_\iota *_{\iota \tarrow \iota \tarrow \iota} z_\iota) =
      (x_\iota *_{\iota \tarrow \iota \tarrow \iota} y_\iota) *_{\iota
        \tarrow \iota \tarrow \iota} z_\iota.$

    \item $\ForallApp x_\iota \mdot e_\iota *_{\iota \tarrow \iota
      \tarrow \iota} x_\iota = x_\iota$.

    \item $\ForallApp x_\iota \mdot x_\iota *_{\iota \tarrow \iota
      \tarrow \iota} e_\iota = x_\iota$.

  \ee

\ee

\noindent
Define $M' = (L_{M'},\Gamma_{M'})$ to be the alternate theory of an
abstract monoid with left and right identity elements where:

\be

  \item $\sB_{M'} =\sB_M$.

  \item $\sC_{M'} = \sC_{\rm log} \cup \set{e_{\iota}^{\rm left},
    e_{\iota}^{\rm right}, *_{\iota \tarrow \iota \tarrow
      \iota}}$. {\sglsp} ($*_{\iota \tarrow \iota \tarrow \iota}$ is
    written as an infix operator.)

  \item $L_{M'}$ is the set of $(\sB_{M'},\sC_{M'})$ expressions.

  \item $\sV_{M'} = \sV_M$.

  \item $\Gamma_{M'}$ contains the following axioms:

  \be

    \item $\ForallApp x_\iota \mdot \ForallApp y_\iota \mdot
      \ForallApp z_\iota \mdot x_\iota *_{\iota \tarrow \iota \tarrow
        \iota} (y_\iota *_{\iota \tarrow \iota \tarrow \iota} z_\iota) =
      (x_\iota *_{\iota \tarrow \iota \tarrow \iota} y_\iota) *_{\iota
        \tarrow \iota \tarrow \iota} z_\iota.$

    \item $\ForallApp x_\iota \mdot e_{\iota}^{\rm left} *_{\iota \tarrow \iota
      \tarrow \iota} x_\iota = x_\iota$.

    \item $\ForallApp x_\iota \mdot x_\iota *_{\iota \tarrow \iota
      \tarrow \iota} e_{\iota}^{\rm right} = x_\iota$.

  \ee

\ee

We would like to construct a semantic morphism from $M'$ to $M$ that
maps the left and right identity elements of $M'$ to the single
identity element of $M$.  This is not possible since the mapping $\nu$
must be injective to overcome the Constant Interpretation Problem.  We
need to add a dummy constant to $M$ to facilitate the definition of
the semantic morphism.  Let $\overline{M}$ be the definitional
extension of $M$ that contains the new constant $e'_\iota$ and the new
axiom $e'_\iota = e_\iota$.\footnote{Technically, $e'_\iota$ is a
  constant chosen from $\sC \setminus \sC_{M}$.  There is no harm is
  assuming that such a constant already exists in $\sC$.}

Let $\Phi = (\mu,\nu)$ to be the translation from $M'$ to
$\overline{M}$ such that:

\be

  \item $\mu(\iota) = \LambdaApp x_\iota
    \mdot T_o$.

  \item $\nu$ is the identity function on $\sV_{M'} \cup \sC_{\log}
    \cup \set{*_{\iota \tarrow \iota \tarrow \iota}}$.

  \item $\nu(e_{\iota}^{\rm left}) = e_\iota$.

  \item $\nu(e_{\iota}^{\rm right}) = e'_\iota$.

\ee

\noindent
It is easy to see that $\Phi$ is a semantic morphism by
Theorem~\ref{thm:sem-morph} .

\subsection{Example 2: Monoid interpreted as the Trivial Monoid}

The identity element of a monoid forms a submonoid of the monoid that
is isomorphic with the trivial monoid consisting of a single element.
There is a natural morphism from a theory of a monoid to itself in
which the type of monoid elements is interpreted by the singleton set
containing the identity element.  This kind of morphism cannot be
directly expressed using a definition of a theory morphism that maps
base types to types.  However, it can be directly expressed using the
notion of a semantic morphism we have defined.

The desired translation interprets the type $\iota$ as the set
$\set{e_\iota}$ and the constants denoting functions involving $\iota$
as functions in which the domain of $\iota$ is replaced by
$\set{e_\iota}$.  This is not possible since the mapping $\nu$ must
map constants to constants to overcome the Constant Interpretation
Problem.  We need to add a set of dummy constants to $M$ to facilitate
the definition of the semantic morphism.

Define $\mu$ as follows:

\be

  \item For $\alpha \in \set{o,\epsilon}$, $\mu(\alpha) = \LambdaApp
    x_\alpha \mdot T_o$.

  \item $\mu(\iota) = \LambdaApp x_\iota \mdot x_\iota = e_\iota$.

\ee

\noindent
Let $\overline{M}$ be the definitional extension of $M$ that contains
the following the new defined constants:

\be

  \item ${='_{\alpha \tarrow \alpha \tarrow o}} = 
    \LambdaApp x_{\alpha} \mdot \LambdaApp y_{\alpha} \mdot
    (\If \; (\overline{\mu}(\alpha) \, x_{\alpha} \And
    \overline{\mu}(\alpha) \, y_{\alpha}) \; (x_{\alpha} =_{\alpha
      \tarrow \alpha \tarrow o} y_{\alpha}) \; \Undefined_o)$\\ where
    $\alpha \in \sT$ contains $\iota$.

  \item ${\iota'_{(\alpha \tarrow o) \tarrow \alpha}} = 
    \LambdaApp x_{\alpha \tarrow o} \mdot 
    (\If \; (\overline{\mu}(\alpha \tarrow o) \, x_{\alpha \tarrow o}) \;
    (\iota_{(\alpha \tarrow o) \tarrow \alpha} \, x_{\alpha \tarrow o}) \;
    \Undefined_{\alpha})$\\
    where $\alpha \in \sT$ contains $\iota$.

  \item $*'_{\iota \tarrow \iota \tarrow \iota} = \LambdaApp x_\iota
    \mdot \LambdaApp y_\iota \mdot (\If \; (\overline{\mu}(\iota) \,
    x_\iota \And \overline{\mu}(\iota) \, y_\iota) \; (x_\iota
    *_{\iota \tarrow \iota \tarrow \iota} y_\iota) \;
    \Undefined_\iota)$.\footnote{The definition of $*'_{\iota \tarrow
        \iota \tarrow \iota}$ can be simplified by using the
      definition of $\mu$ and noting that $e_\iota *_{\iota \tarrow
        \iota \tarrow \iota} e_\iota$ equals $e_\iota$.}

\ee

Let $\Psi = (\mu,\nu)$ to be the translation from $M$ to
$\overline{M}$ such that:

\be

  \item $\mu$ is defined as above.

  \item $\nu$ is the identity function on $\sV_{M'}$.

  \item $\nu$ is the identity function on the members of $\sC_{\rm
    log}$ except for the constants $=_{\alpha \tarrow \alpha \tarrow
    o}$ and $\iota_{(\alpha \tarrow o) \tarrow \alpha}$ where $\alpha
    \in \sT$ contains $\iota$.

  \item ${\nu(=_{\alpha \tarrow \alpha \tarrow o})} = {='_{\alpha
      \tarrow \alpha \tarrow o}}$ for all $\alpha \in \sT$ containing
    $\iota$.

  \item ${\nu(\iota_{(\alpha \tarrow o) \tarrow \alpha})} =
    {\iota'_{(\alpha \tarrow o) \tarrow \alpha}}$ for all $\alpha \in
    \sT$ containing $\iota$.

  \item $\nu(e_\iota) = e_\iota$.

  \item $\nu(*_{\iota \tarrow \iota \tarrow \iota}) = *'_{\iota \tarrow
    \iota \tarrow \iota}$.

\ee

\noindent
It is easy to see that $\Phi$ is a semantic morphism by
Theorem~\ref{thm:sem-morph} .

\section{Conclusion}\label{sec:conclusion}

{\churchqe} is a version of Church's type theory with quotation and
evaluation described in great detail in~\cite{FarmerArxiv16}.  In this
paper we have (1) presented {\churchuqe}, a variant of {\churchqe}
that admits undefined expressions, partial functions, and multiple
base types of individuals, (2) defined a notion of a theory morphism
in {\churchuqe}, and (3) given two simple examples that illustrate the
use of theory morphisms in {\churchuqe}.  The theory morphisms of
{\churchuqe} overcome the Constant Interpretation Problem discussed in
section~\ref{sec:introduction} by requiring constants to be
injectively mapped to constants.  Since {\churchuqe} admits partial
functions, {\churchuqe} theory morphisms are able to map base types to
sets of values of the same type --- which enables many additional
natural meaning-preserving mappings between theories to be directly
defined as {\churchuqe} theory morphisms.  Thus the paper demonstrates
how theory morphisms can be defined in a traditional logic with
quotation and evaluation and how support for partial functions in a
traditional logic can be leveraged to obtain a wider class of theory
morphisms.

The two examples presented in section~\ref{sec:examples} show that
constructing a translation in {\churchuqe} from a theory $T_1$ to a
theory $T_2$ will often require defining new dummy constants in $T_2$.
This is certainly a significant inconvenience.  However, it is an
inconvenience that can be greatly ameliorated in an implementation of
{\churchuqe} by allowing a user to define a ``pre-translation'' that
is automatically transformed into a bona fide translation.  A
pre-translation from $T_1$ and $T_2$ would be a pair $(\mu,\nu)$ where
$\mu$ maps base types to either types or semantically closed
predicates, $\nu$ maps constants to expressions that need not be
constants, and $\nu$ is not required to be injective.  From the
pre-translation, the system would automatically extend $T_2$ to a
theory $T'_2$ and then construct a translation from $T_1$ to $T'_2$.

Our long-range goal is to implement a system for developing biform
theory graphs utilizing logics equipped with quotation and evaluation.
The next step in this direction is to implement {\churchqe} by
extending HOL Light~\cite{Harrison09}, a simple implementation of
HOL~\cite{GordonMelham93}.

\bibliography{$HOME/research/lib/imps}%$ 
\bibliographystyle{plain}

\end{document}